\documentclass[letterpaper, 10 pt, conference]{ieeeconf}  

\IEEEoverridecommandlockouts                              
\usepackage{amsmath,amssymb,url,listings,mathrsfs,wasysym}
\usepackage{enumerate}
\usepackage{graphicx}
\usepackage{subfigure}
\usepackage{cite}
\usepackage[linesnumbered,vlined,ruled]{algorithm2e}

\newcommand{\num}[1]{\relax\ifmmode \mathbb #1\else $\mathbb #1$\fi}
\newcommand{\naturals}{{\num N}^0}
\newcommand{\reals}{{\num R}}
\newcommand{\expect}{\mathbb{E}}
\newcommand{\conv}{\mathit{Conv}}

\newcommand{\controller}{\mathcal{C}}

\newcommand{\J}{{\mathcal{J}}}
\newcommand{\R}{{\mathcal{L}}}
\newcommand{\M}{{\mathcal{M}}}
\newcommand{\N}{{\mathcal{N}}}
\renewcommand{\P}{{\mathcal{P}}}
\newcommand{\Z}{{\mathcal{Z}}}

\newtheorem{theorem}{Theorem}
\newtheorem{lemma}{Lemma}
\newtheorem{proposition}{Proposition}
\newtheorem{corollary}{Corollary}
\newtheorem{definition}{Definition}

\title{\LARGE \bf
Distributed Traffic Signal Control for Maximum Network Throughput
}
\author{Tichakorn Wongpiromsarn, Tawit Uthaicharoenpong, Yu Wang, Emilio Frazzoli and Danwei Wang\vspace{-5mm}
\thanks{T. Wongpiromsarn is with the Singapore-MIT Alliance for Research and Technology,
Singapore {\tt\footnotesize nok@smart.mit.edu}}
\thanks{T. Uthaicharoenpong, Y. Wang and D. Wang are with the Nanyang Technological University, Singapore
{\tt \footnotesize tawit@ntu.edu.sg, wangyu81@e.ntu.edu.sg, edwwang@ntu.edu.sg}}
\thanks{E.  Frazzoli is with the Massachusetts Institute of Technology, Cambridge, MA, USA
  {\tt\footnotesize frazzoli@mit.edu}}%
}
\begin{document}
\maketitle
\thispagestyle{empty}
\pagestyle{empty}

\begin{abstract}
We propose a distributed algorithm for controlling traffic signals.
Our algorithm is adapted from backpressure routing, which has been mainly applied to communication
and power networks.
We formally prove that our algorithm ensures global optimality as it leads to maximum network throughput even though
the controller is constructed and implemented in a completely distributed manner.
Simulation results show that our algorithm significantly outperforms SCATS, 
an adaptive traffic signal control system that is being used in many cities.
\end{abstract}

\section{Introduction}
Traffic signal control is a key element in traffic management that affects the efficiency of urban transportation.
Many major cities worldwide currently employ adaptive traffic signal control systems where the light timing is adjusted based
on the current traffic situation.
Examples of widely-used adaptive traffic signal control systems include SCATS (Sydney Coordinated Adaptive Traffic System)~\cite{Lowrie82SCATS,Keong93Glide,Liu03Master} and 
SCOOT (Split Cycle Offset Optimisation Technique)~\cite{Day1998,Stevanovic2008}.

Control variables in traffic signal control systems typically include phase, cycle length, split plan and offset.
A phase specifies a combination of one or more traffic movements simultaneously receiving the right of way during a signal interval.
Cycle length is the time required for one complete cycle of signal intervals.
A split plan defines the percentage of the cycle length allocated to each of the phases during a signal cycle.
Offset is used in coordinated traffic control systems to reduce frequent stops at a sequence of junctions.

SCATS, for example, attempts to equalize the degree of saturation (DS), 
i.e., the ratio of effectively used green time to the total green time,
for all the approaches.
The computation of cycle length and split plan is only carried out at the critical junctions.
Cycle length and split plan at non-critical junctions are controlled by the critical junctions via offsets.
The algorithm involves many parameters, which need to be properly calibrated for each critical junction.
In addition, all the possible split plans need to be pre-specified and
a voting scheme is used in order to select a split plan that leads to approximately equal DS for all the approaches.

Systems and control theory has been recently applied to traffic signal control problems.
In \cite{Diakaki02}, a multivariable regulator is proposed based on linear-quadratic regulator methodology and
the store-and-forward modeling approach \cite{Aboudolas08}.
Robust control theory has been applied to traffic signalization in \cite{Yu97thesis}.
Approaches based on Petri Net modeling language are considered in, e.g., \cite{Mladenovic11thesis,Soares08}.
Optimization-based techniques are considered, e.g., in \cite{Dujardin11,Shen11}.
However, one of the major drawbacks of these approaches is the scalability issue, which limits
their application to relatively small networks.

To address the scalability issue, in \cite{Cheng09}, a distributed algorithm is presented where 
the signal at each junction is locally controlled independently from other junctions.
However, global optimality is no longer guaranteed, although simulation results show that it reduces
the total delay compared to the fixed-time approach.
Another distributed approach is considered in \cite{Lammer08} 
where the constraint that each traffic flow is served once, on average, within a desired service interval $T$ is imposed.
It can be proved that their distributed algorithm
stabilizes the network whenever there exists a stable fixed-time control with cycle time $T$.
However, the knowledge of traffic arrival rates is required.
In addition, multi-phase operation is not considered.

An objective of this work is to develop a traffic signal control strategy that requires minimal tuning
and scales well with the size of the road network while ensuring satisfactory performance.
Our algorithm is motivated by backpressure routing introduced in \cite{Tassiulas92Stability},
which has been mainly applied to communication and power networks 
where a packet may arrive at any node in the network and can only leave the system when it reaches its destination node.
One of the attractive features of backpressure routing is that it leads to maximum network throughput 
without requiring any knowledge about traffic arrival rates
\cite{Tassiulas92Stability,NMR05,Georgiadis06}.

To the authors' knowledge, this is the first time backpressure routing has been adapted
to solve the traffic signal control problem.
Since many assumptions made in backpressure routing are not valid in our traffic signalization application,
certain modifications need to be made to the original algorithm.
With these modifications, we formally prove that our algorithm inherits the desired properties of backpressure routing
as it leads to maximum network throughput even though the signal at each junction is determined completely independently
from the signal at other junctions, and no information about traffic arrival rates is provided.
Furthermore, since our controller is constructed and implemented in a completely distributed manner, 
it can be applied to an arbitrarily large network.
Simulation results show that our algorithm significantly outperforms SCATS.

The remainder of the paper is organized as follows: 
We provide useful definitions and existing results concerning network stability in the following section. 
Section \ref{sec:prob} describes the traffic signal control problem considered in this paper.
Our backpressure-based traffic signal control algorithm is described in Section \ref{sec:contr}.
In Section \ref{sec:contr_evaluation}, we formally prove that our algorithm
ensures global optimality as it leads to maximum network throughput, even though the signal at each junction
is determined completely independently from other junctions.
Section \ref{sec:results} presents simulation results, showing that our algorithm can significantly reduce
the queue length compared to SCATS.
Finally, Section \ref{sec:conclusions} concludes the paper and discusses future work.

\section{Preliminaries}
In this section, we summarize existing results and definitions concerning network stabilility.
We refer the reader to \cite{Tassiulas92Stability,NMR05,Georgiadis06} for more details.

Consider a network modeled by a directed graph with $N$ nodes and $L$ links.
Each node maintains an internal queue of objects to be processed by the network, while
each link $(a,b)$ represents a channel for direct transmission of objects from node $a$ to node $b$.
Suppose the network operates in slotted time $t \in \naturals$ where $\naturals$ is the set of natural numbers (including zero).
Objects may arrive at any node in the network and can only leave the system upon reaching the their destination node.
Let $A_i(t)$ represent the number of objects that exogenously arrives at source node $i$ during slot $t$ and
$U_i(t)$ represent the queue length at node $i$ at time $t$.
We assume that all the queues have infinite capacity.
In addition, only the objects currently at each node at the beginning of slot $t$ can be transmitted during that slot.
Our control objective is to ensure that all queues are stable as defined below. 

\begin{definition}
A network is \emph{strongly stable} if each individual queue $U$ satisfies
\begin{equation}
\limsup_{t \to \infty} \frac{1}{t} \sum_{\tau=0}^{t-1} 1_{[U(\tau) > V]} \to 0 \hbox{ as } V \to \infty,
\end{equation}
where for any event $X$, the indicator function $1_X$ takes the value 1 if X is satisfied 
and takes the value 0 otherwise.
\end{definition}


In this paper, we restrict our attention to strong stability and use the term ``stability'' to refer to strong stability defined above.
%
%
For a network with $N$ queues $U_1, \ldots, U_N$ that evolve according to some probabilistic law, 
a sufficient condition for stability can be provided using Lyapunov drift.

\begin{proposition}
\label{prop:LyapunovStability}
Suppose $\expect\{U_i(0)\} < \infty$ for all $i \in \{1, \ldots, N\}$
and there exist constants $B > 0$ and $\epsilon > 0$ such that
\begin{equation}
\expect \Big\{L(\mathbf{U}(t+1)) - L(\mathbf{U}(t)) \Big| \mathbf{U}(t)\Big\} \leq B - \epsilon \sum_{i=1}^N U_i(t), \forall t \in \naturals,
\end{equation}
where for any queue vector $\mathbf{U} = [U_1, \ldots, U_N]$, $L(\mathbf{U}) \triangleq \sum_{i=1}^N U_i^2$.
Then the network is strongly stable.
\end{proposition}

\begin{definition}
\label{def:rate}
An arrival process $A(t)$ is \emph{admissible with rate} $\lambda$ if:
\begin{itemize}
\item The time average expected arrival rate satisfies
\begin{equation*}
\lim_{t \to \infty} \frac{1}{t} \sum_{\tau=0}^{t-1} \expect\{A(\tau)\} = \lambda.
\end{equation*}
\item There exists a finite value $A_{max}$ such that $\expect\{A(t)^2 \hspace{1mm}|\hspace{1mm} \mathbf{H}(t)\} \leq A_{max}^2$
for any time slot $t$, where $\mathbf{H}(t)$ represents the history up to time $t$, i.e., all events that take place during slots $\tau \in \{0,\ldots,t-1\}$.
\item For any $\delta > 0$, there exists an interval size $T$ (which may depend on $\delta$) such that for any initial time $t_0$,
\begin{equation*}
\expect\left\{ \frac{1}{T}  \sum_{k=0}^{T-1} A(t_0 + k) \hspace{1mm}\Big|\hspace{1mm} \mathbf{H}(t_0) \right\} \leq \lambda + \delta.
\end{equation*}
\end{itemize}
\end{definition}

For each node $i$, we define $\lambda_i$ to be the time average rate with which 
$A_i(t)$ is admissible.
Let $\boldsymbol{\lambda} = \left[\lambda_i\right]$ represent the arrival rate vector.


\begin{definition}
\label{def:capacity region}
The \emph{capacity region} $\Lambda$ is the closed region of arrival rate vectors $\boldsymbol{\lambda}$
with the following properties:
\begin{itemize}
\item $\boldsymbol{\lambda} \in \Lambda$ is a necessary condition for network stability, considering
all possible strategies for choosing the control variables 
(including strategies that have perfect knowledge of future events).
\item $\boldsymbol{\lambda} \in \mathrm{int}(\Lambda)$ 
is a sufficient condition for the network to be stabilized by a
policy that does not have a-priori knowledge of future events.
\end{itemize}
\end{definition}

The capacity region essentially describes the set of all arrival rate vectors that can be stably supported by the network.
A scheduling algorithm is said to maximize the network throughput if it stabilizes 
the network for all arrival rates in the interior of $\Lambda$.

\section{The Traffic Signal Control Problem}
\label{sec:prob}
A road network $\N$ is defined as a collection of links and signalized junctions.
Let $N$ and $L$ be the number of links and junctions, respectively, in $\N$.
Then, $\N$ can be written as 
$\N = (\R, \J)$ where $\R = \{\R_1, \ldots, \R_{N}\}$
and $\J = \{ \J_1, \ldots, \J_{L}\}$ are sets of all the links and signalized junctions, respectively, in $\N$.
Each junction $\J_i$ can be described by a tuple $\J_i = (\M_i, \P_i, \Z_i)$ where 
$\M_i \subseteq \R^2$ is a set of all the possible traffic movements through $\J_i$, 
$\P_i \subseteq 2^{\M_i}$ is a set of all the possible phases of $\J_i$ and
$\Z_i$ is a finite set of traffic states, each of which captures factors
that affect the traffic flow rate through $\J_i$ such as traffic and weather conditions.
Each traffic movement through junction $\J_i$ is defined by a pair $(\R_a, \R_b)$ where $\R_a, \R_b \in \R$ 
such that a vehicle may enter and exit $\J_i$ through $\R_a$ and $\R_b$, respectively.
Each phase $p \in \P_i$ defines a combination $p \subseteq \M_i$ of traffic movements simultaneously receiving the right-of-way.
A typical set of phases of a 4-way junction is shown in Figure \ref{fig:phase-ex}.

\begin{figure}[h] 
   \centering 
   \subfigure[]
        {\includegraphics[trim=5.5cm 10cm 13cm 4cm, clip=true, width=0.2\textwidth]{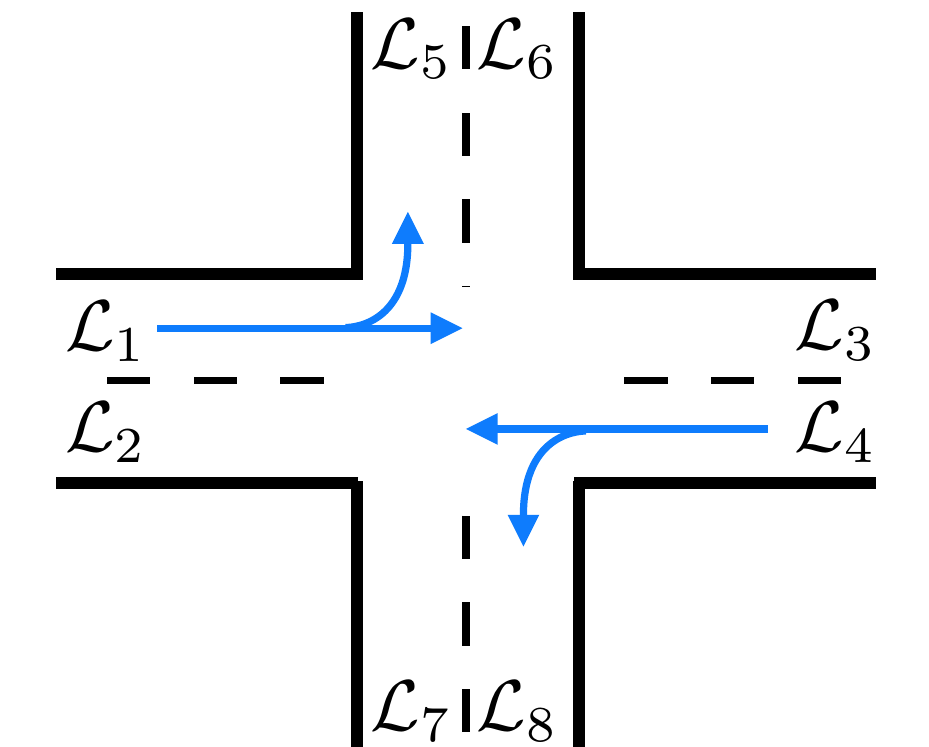}}
   \subfigure[]
        {\includegraphics[trim=5.5cm 10cm 13cm 4cm, clip=true, width=0.2\textwidth]{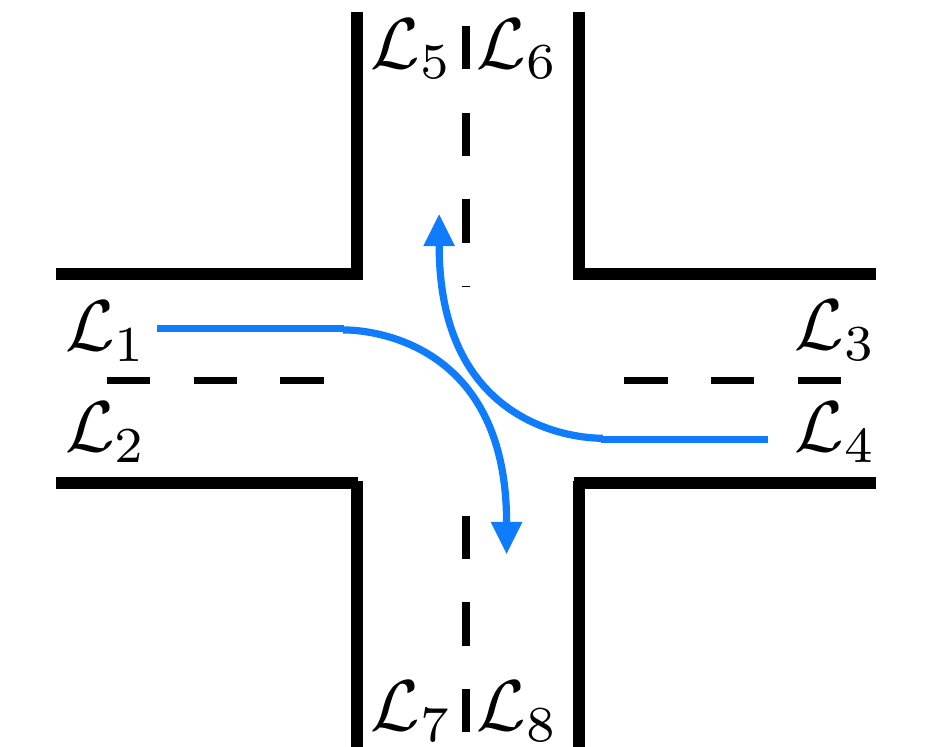}}
   \subfigure[]
        {\includegraphics[trim=5.5cm 10cm 13cm 4cm, clip=true, width=0.2\textwidth]{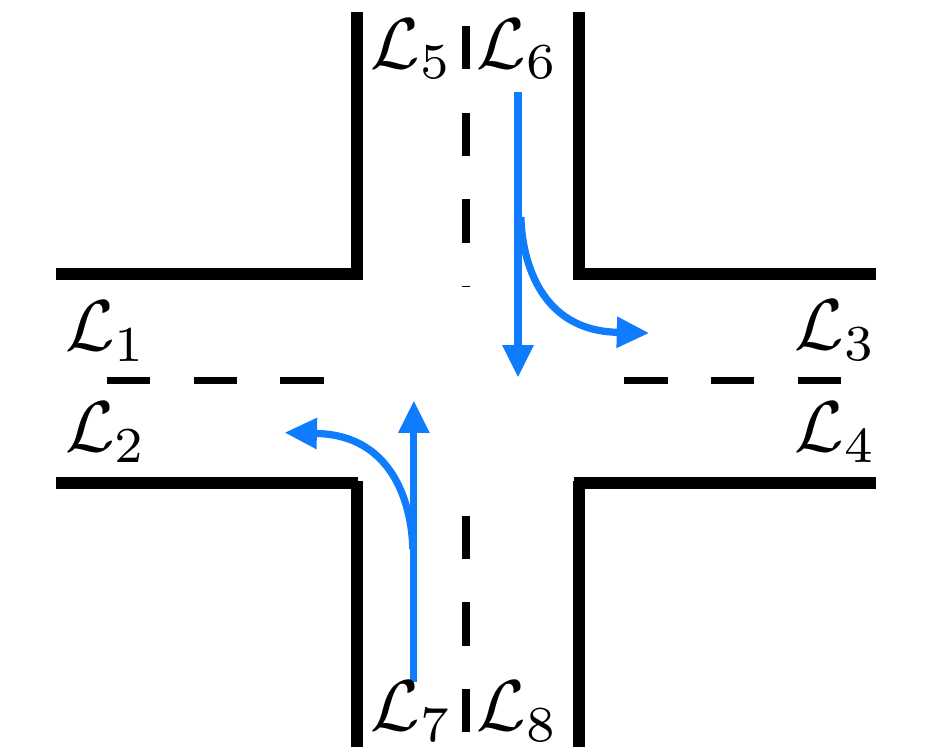}}
   \subfigure[]
        {\includegraphics[trim=5.5cm 10cm 13cm 4cm, clip=true, width=0.2\textwidth]{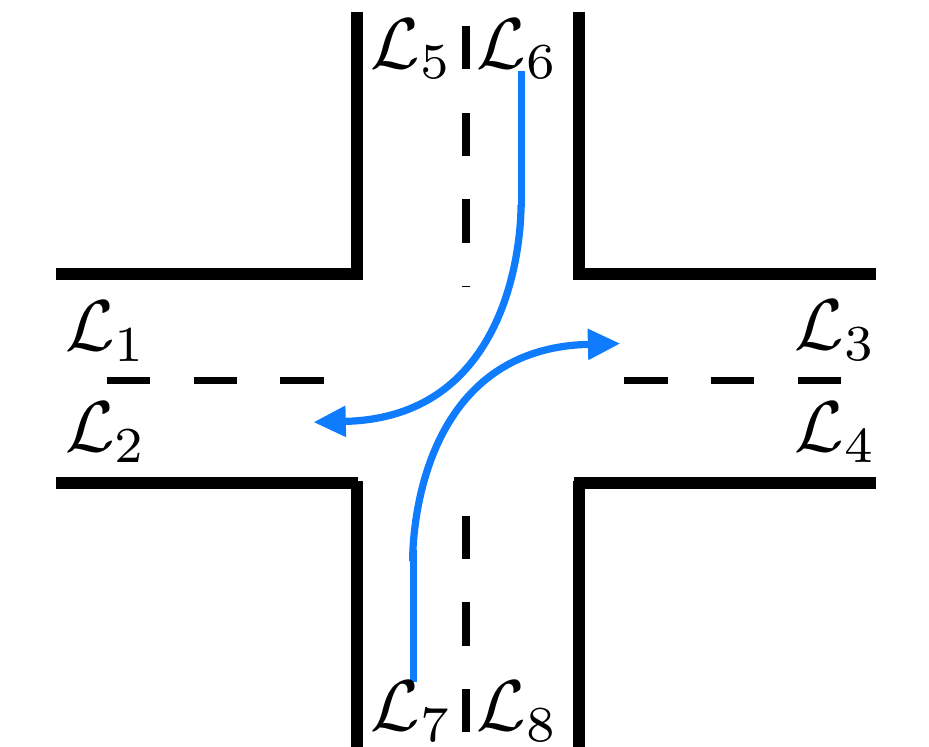}}
   \caption{A typical set $\{\P_1, \P_2, \P_3, \P_4\}$ of phases of a 4-way junction with links $\R_1, \ldots, \R_8$.
   (a) $\P_1 = \{(\R_1, \R_3), (\R_1, \R_5), (\R_4, \R_2), (\R_4, \R_8)\}$,
   (b) $\P_2 = \{(\R_1, \R_8), (\R_4, \R_5)\}$,
   (c) $\P_3 = \{(\R_7, \R_5), (\R_7, \R_2)$, $(\R_6, \R_8), (\R_6, \R_3)\}$, and
   (d) $\P_4 = \{(\R_7, \R_3), (\R_6, \R_2)\}$.}
  \vspace{-5mm}
  \label{fig:phase-ex}
\end{figure}

We assume that the traffic signal system operates in slotted time $t \in \naturals$.
During each time slot, vehicles may enter the network at any link. 
For each $a \in \{1, \ldots, N\}$, $i \in \{1, \ldots, L\}$, $t \in \naturals$, 
we let $Q_a(t) \in \naturals$ and $z_i(t) \in \Z_i$ represent the number of vehicles on $\R_a$
and the traffic state around $\J_i$, respectively, at the beginning of time slot $t$.
In addition, for each $i \in \{1, \ldots, L\}$,
we define a function $\xi_{i} : \P_i \times \M_i \times \Z_i \to \naturals$ such that
$\xi_{i}(p, \R_a, \R_b, z)$ gives the rate (i.e., the number of vehicles per unit time) 
at which vehicles that can go from $\R_a$ to $\R_b$ through junction $\J_i$ 
under traffic state $z$ if phase $p$ is activated.
By definition, $\xi_{i}(p, \R_a, \R_b, z) = 0, \forall z \in \Z_i$ if $(\R_a, \R_b) \not\in p$, i.e., 
phase $p$ does not give the right of way to the traffic movement from $\R_a$ to $\R_b$.
When traffic state $z$ represents the case where the number of vehicles on $\R_a$ that 
seek the movement to $\R_b$ through $\J_i$ is large, 
$\xi_{i}(p, \R_a, \R_b, z)$ can be simply obtained by assuming saturated flow.

At the beginning of each time slot, the traffic signal controller determines the phase for each junction to be
activated during this time slot.
In this paper, we consider the traffic signal control problem as stated below.

\textbf{Traffic Signal Control Problem:} Design a traffic signal controller that determines the phase $p_i(t) \in \P_i$
for each junction $\J_i, i \in \{1, \ldots, L\}$ to be activated during each time slot $t \in \naturals$ such that the network throughput is maximized.
We assume that there exists a reliable traffic monitoring system that provides the queue length $Q_a(t)$ 
and traffic state $z_i(t)$ for each $a \in \{1, \ldots, N\}$, $i \in \{1, \ldots, L\}$  
at the beginning of each time slot $t \in \naturals$ to the controller.

\section{Backpressure-based Traffic Signal Controller}
\label{sec:contr}

In this section, we propose a distributed traffic signal control algorithm 
that employs the idea from backpressure routing 
as described in \cite{Tassiulas92Stability,NMR05,Georgiadis06}.
Unlike most of the traffic signal controllers considered in existing literature, 
our controller can be constructed and implemented in a completely distributed manner.
Furthermore, it does not require any knowledge about traffic arrival rates.
We end the section with a discussion of some basic properties of the proposed controller.

Our traffic signal controller consists of a set of local controllers $\controller_1, \ldots, \controller_L$
where local controller $\controller_i$ is associated with junction $\J_i$.
These local controllers are constructed and implemented independently%
\footnote{However, a synchronized operation among all the junctions is required so that 
control actions for all the junctions take place according to a common time clock.}
of one another.
Furthermore, each local controller does not require the global view of the road network.
Instead, it only requires information that is local to the junction with which it is associated.
At each time slot $t$, local controller $\controller_i$ computes the phase $p^* \in \P_i$
to be activated at junction $\J_i$ during time slot $t$ as described in Algorithm \ref{Alg:controller}.

\begin{algorithm}[t]
\SetInd{0.5em}{0.5em}
\KwIn{$z_i(t)$ and $Q_a(t)$ for all $a \in \{1, \ldots, N\}$ such that $(\R_a, \R_b) \in \M_i$ or $(\R_b, \R_a) \in \M_i$ for some $\R_b \in \R$}
\KwOut{$p^* \in \P_i$ to be activated during time slot $t$}
$S_p^* \leftarrow -\infty$\;
$p^* \leftarrow \emptyset$\;
\ForEach{$(\R_a, \R_b) \in \M_i$}{
$W_{ab} \leftarrow Q_a(t) - Q_b(t)$\;
}
\ForEach{$p \in \P_i$}{
  $S_p \leftarrow \sum_{(\R_a, \R_b) \in p} W_{ab} \xi_i(p, \R_a, \R_b, z_i(t))$\;
  \If {$S_p > S_p^*$}{
    $p^* \leftarrow p$\;
    $S_p^* = S_p$\;
  }
}
\caption{Computation of phase $p^*$ to be activated during time slot $t$ at junction $\J_i$.}
\label{Alg:controller}
\end{algorithm}

Consider an arbitrary junction $\J_i \in \J$.
At the beginning of time slot $t$, we first compute (line 4 of Algorithm \ref{Alg:controller})
\begin{equation}
  \label{eq:Wab_traffic}
  W_{ab}(t) \triangleq Q_a(t) -  Q_b(t), 
\end{equation}
for each pair $(\R_a, \R_b) \in \M_i$.
Then, for each phase $p \in \P_i$, we compute (line 6 of Algorithm \ref{Alg:controller})
\begin{equation}
  \label{eq:Sp}
  S_p(t) \triangleq \sum_{(\R_a, \R_b) \in p} W_{ab}(t) \xi_{i}(p, \R_a, \R_b, z_i(t)).
\end{equation}

The local controller $\controller_i$ then activates phase $p^* \in \P_i$ such that $S_{p^*} \geq S_p, \forall p \in \P_i$ during the time slot $t$
(line 7--9 of Algorithm \ref{Alg:controller}).
If there exist multiple options of $p^*$ that satisfy the inequality, the controller can pick one arbitrarily.
Note that since the number of possible phases for each junction is typically small (e.g., less than $10$),
the above computation and enumeration through all the possible phases can be practically performed in real time.

Our algorithm is similar in nature to backpressure routing for a single-commodity network.
In \cite{Tassiulas92Stability,NMR05,Georgiadis06}, it has been shown that backpressure routing
leads to maximum network throughput.
However, it is still premature to simply conclude that our backpressure-based traffic signal control algorithm 
inherits this property due to the following reasons.
First, backpressure routing requires that a commodity at least defines the destination of the object.
Implementing the algorithm for a single-commodity network implies that we assume that all the vehicles have a common destination,
which is not a valid assumption for our application.
Second, backpressure routing assumes that the controller has complete control over routing of the traffic around the network whereas
in our traffic signal control problem, the controller does not have control over the route picked by each driver.
Third, backpressure routing assumes that the network controller has control over the flow rate
of each link subject to the maximum rate imposed by the link constraint.
However, the traffic signal controller can only picks a phase $p_i(t)$ to be activated at each junction $\J_i$ during each time slot $t$ but 
does not have control over the flow rate of each traffic movement once $p_i(t)$ is activated.
To account for this lack of control authority, we slightly modify the definition of $W_{ab}(t)$ from
that used in backpressure routing.
Finally, the optimality result of backpressure routing relies on the assumption that
all the queues have infinite buffer storage space.
Even though it is not reasonable to assume that all the links have infinite queue capacity,
for the rest of the paper, we assume that this is the case.
In practice, our algorithm is expected to work well when each link can accommodate a reasonably long queue.

Before evaluating the performance of our algorithm, 
we first provide its basic property, which is similar to the basic property of backpressure routing. 
Let $\P = \P_1 \times \ldots \times \P_L$ and $\Z = \Z_1 \times \ldots \times \Z_L$.
For each $a \in \{1, \ldots, N\}$, we define functions
$V^{out}_a : \P \times \Z \to \reals$ and $V^{in}_a : \P \times \Z \to \reals$ 
such that for any $\mathbf{p} \in \P$ and $\mathbf{z} \in \Z$,
\begin{equation}
\begin{array}{rcl}
V^{out}_a(\mathbf{p}, \mathbf{z}) &=&  
\displaystyle{\sum_{\scriptsize \begin{array}{c}b,i \hbox{ s.t. }\\ (\R_a, \R_b) \in \M_i \end{array}} \hspace{-6mm} \xi_{i}(p_i, \R_a, \R_b, z_i)},\\
V^{in}_a(\mathbf{p}, \mathbf{z}) &=& 
\displaystyle{\sum_{\scriptsize \begin{array}{c}b,i \hbox{ s.t. }\\ (\R_b, \R_a) \in \M_i \end{array}} \hspace{-6mm} \xi_{i}(p_i, \R_b, \R_a, z_i)},
\end{array}
\end{equation}
where for each $i \in \{1, \ldots, L\}$,
$p_i \in \P_i$ is the element of $\mathbf{p}$ that corresponds to the phase of junction $\J_i$ and
$z_i \in \Z_i$ is the element of $\mathbf{z}$ that corresponds to the traffic state of junction $\J_i$.

\begin{lemma}
\label{lem:basic_prop}
Consider an arbitrary time slot $t \in \naturals$.
Let $\mathbf{z}(t) \in \Z$ be a vector of traffic states of
all the junctions during time slot $t$.
For each $i \in \{1, \ldots, L\}$, 
let $p_i^*(t)$ denote the phase determined by Algorithm \ref{Alg:controller} to be activated at junction $\J_i$ during time slot $t$
and $\tilde{p}_i(t)$ be the phase to be activated at junction $\J_i$ determined by any other algorithm for junction $\J_i$ during time slot $t$.
Then, 
\begin{equation}
\label{eq:lem:basic_prop}
\begin{array}{l}
\displaystyle{\sum_a Q_a(t) \Big( V^{out}_a \big(\tilde{\mathbf{p}}(t), \mathbf{z}(t) \big) - V^{in}_a \big(\tilde{\mathbf{p}}(t), \mathbf{z}(t) \big) \Big)} \\
\hspace{8mm}\leq
\displaystyle{\sum_a Q_a(t) \Big( V^{out}_a \big(\mathbf{p}^*(t), \mathbf{z}(t) \big) - V^{in}_a \big(\mathbf{p}^*(t), \mathbf{z}(t) \big) \Big)},
\end{array}
\end{equation}
where $\tilde{\mathbf{p}}(t) = [\tilde{p}_i(t)]$ and $\mathbf{p}^*(t) = [p^*_i(t)]$.
\end{lemma}
\begin{proof}
First, we note the following identity
\begin{equation}
\label{pf:basic_prop1}
\begin{array}{l}
\displaystyle{\sum_a Q_a(t) \Big( V^{out}_a\big(\mathbf{p}(t), \mathbf{z}(t)\big) - V^{in}_a\big(\mathbf{p}(t), \mathbf{z}(t)\big) \Big)}\\
\hspace{8mm}=\hspace{-5mm}
\displaystyle{\sum_{\scriptsize \begin{array}{c}a,b,i \hbox{ s.t. }\\ (\R_a, \R_b) \in \M_i \end{array}} \hspace{-6mm} \xi_{i} \big(p_i(t), \R_a, \R_b, z_i(t) \big)
W_{ab}(t)},
\end{array}
\end{equation}
for all $\mathbf{p}(t) \in \P$ and $\mathbf{z}(t) \in \Z$.

Since for each $i \in \{1, \ldots, L\}$, $p^*_i(t)$ is chosen such that $S_{p^*_i(t)} \geq S_{\tilde{p}_i(t)}$,
we get
\begin{equation}
\label{pf:basic_prop2}
\begin{array}{l}
\displaystyle{\sum_{\scriptsize \begin{array}{c}a,b \hbox{ s.t. }\\ (\R_a, \R_b) \in \M_i \end{array}} \hspace{-6mm} \xi_i \big(\tilde{p}_i(t), \R_a, \R_b, z_i(t) \big)
W_{ab}(t)}\\
\hspace{8mm}\leq
\displaystyle{\sum_{\scriptsize \begin{array}{c}a,b \hbox{ s.t. }\\ (\R_a, \R_b) \in \M_i \end{array}} \hspace{-6mm} \xi_i \big(p^*_i(t), \R_a, \R_b, z_i(t) \big)
W_{ab}(t)},
\end{array}
\end{equation}
for all $i \in \{1, \ldots, L\}$.
The result in (\ref{eq:lem:basic_prop}) can be obtained by summing the inequality in (\ref{pf:basic_prop2}) over
$i \in \{1, \ldots, L\}$ and using the identity in (\ref{pf:basic_prop1}).
\end{proof}

\section{Controller Performance Evaluation}
\label{sec:contr_evaluation}

Let $\Lambda$ be the capacity region of the road network as defined in Definition \ref{def:capacity region}.
Assume that $\mathbf{z}(t) = \left[ z_i(t) \right]$ evolve according to a finite state, irreducible, aperiodic Markov chain.
Let $\pi_\mathbf{z}$ represent the time average fraction of time that $\mathbf{z}(t) = \mathbf{z}$, i.e., with probability 1, we have
$\lim_{t \to \infty} \frac{1}{t} \sum_{\tau = 0}^{t-1} 1_{[\mathbf{z}(\tau) = \mathbf{z}]} = \pi_\mathbf{z}$, for all 
$\mathbf{z} \in \Z$
where $1_{[\mathbf{z}(\tau) = \mathbf{z}]}$ is an indicator function that takes the value 1 if $\mathbf{z}(\tau) = \mathbf{z} $ 
and takes the value 0 otherwise.
In addition, we let $\M = \bigcup_i \M_i$ be the set of all the possible traffic movements. 
For the simplicity of the presentation, we assume that $\M_i \cap \M_j = \emptyset$ for all $i \not= j$.
For each $\mathbf{p} \in \P$, $\mathbf{z} \in \Z$,
we define a vector $\boldsymbol{\xi}(\mathbf{p}, \mathbf{z})$ whose $k^{th}$ element is equal to
$\xi_i(p_i, \R_a, \R_b, z_i)$ where $(\R_a, \R_b)$ is the $k^{th}$ traffic movement in $\M$,
$i$ is the (unique) index satisfying $(\R_a, \R_b) \in \M_i$ and 
$p_i$ and $z_i$ are the $i^{th}$ element of $\mathbf{p}$ and $\mathbf{z}$, respectively.
Define
\begin{equation}
\label{eq:Gamma}
\Gamma \triangleq \sum_{\mathbf{z} \in \Z} \pi_\mathbf{z} 
\conv \Big\{ \left[ \boldsymbol{\xi}(\mathbf{p}, \mathbf{z}) \right] 
\hspace{1mm}\Big|\hspace{1mm} \mathbf{p} \in \P_1 \times \ldots \times \P_L \Big\},
\end{equation}
where for any set $\mathcal{S}$, $\conv\{\mathcal{S}\}$ represents the convex hull of $\mathcal{S}$.

Additionally, we assume that the process of vehicles exogenously entering the network is rate ergodic
and for all for all $a \in \{1, \ldots, N\}$, there are always enough vehicles on $\R_a$ such that 
for all $i \in \{1, \ldots, L\}$, $b \in \{1, \ldots, N\}$, $p \in \P_i$, $z \in \Z_i$ such that $(\R_a, \R_b) \in \M_i$, 
vehicles can move from $\R_a$ to $\R_b$ through junction $\J_i$ 
at rate  $\xi_i(p, \R_a, \R_b, z)$ under traffic state $z$ if phase $p$ is activated at $\J_i$.
For each $a \in \{1, \ldots, N\}$, let $\lambda_a$ be the time average rate with which the number of new vehicles that exogenously enter
the network at link $\R_a$  during each time slot is admissible.
Let $\boldsymbol{\lambda} = \left[ \lambda_a \right]$ represent the arrival rate vector.

Before deriving the optimality result for our backpressure-based traffic signal control algorithm,
we first characterize the capacity region of the road network, as formally stated in the following lemma.

\begin{lemma}
The capacity region of the network is given by the set $\Lambda$ consisting of
all the rate vectors $\boldsymbol{\lambda}$ such that
there exists a rate vector $\mathbf{G} \in \Gamma$ together with flow variables
$f_{ab}$ for all $a,b \in \{1, \ldots, N\}$ satisfying
\begin{eqnarray}
\label{eq:f-nonnegative}
f_{ab} \geq 0, &&\forall a,b \in \{1, \ldots, N\},\\
\label{eq:f-conservation}
\lambda_a = \sum_b f_{ab} - \sum_c f_{ca}, &&\forall a \in \{1, \ldots, N\},\\
\label{eq:f-zero}
f_{ab} = 0, &&\forall a,b \in \{1, \ldots, N\} \\
\nonumber
&& \hbox{such that } (\R_a, \R_b) \not\in \M,
\end{eqnarray}
\begin{eqnarray}
\label{eq:f-constraint}
f_{ab} = G_{ab}, &&\forall a,b \in \{1, \ldots, N\}  \\
\nonumber
&& \hbox{such that } (\R_a, \R_b) \in \M,
\end{eqnarray}
where $G_{ab}$ is the element of $\mathbf{G}$ that corresponds to the rate of traffic movement
$(\R_a, \R_b)$.
\end{lemma}
\begin{proof}
%
First, we prove that $\boldsymbol{\lambda} \in \Lambda$ is a necessary condition for network stability, 
considering all possible strategies for choosing the control variables (including strategies that have perfect knowledge of future events).
Consider an arbitrary time $t$.
For each $a \in \{1, \ldots, N\}$, let $X_a(t)$ denote the total number of vehicles that exogenously enters the road network
at link $\R_a$ during time interval $[0,t]$.
Suppose the network can be stabilized by some policy,
possibly one that bases its decisions upon complete knowledge of future arrivals.
For each $a,b \in \{1, \ldots, N\}$, let $Q_a(t)$ and $F_{ab}(t)$ represent the number of vehicles left on $\R_a$ at time $t$ 
and the total number of vehicles executing the $(\R_a, \R_b)$ movement during time interval $[0,t]$ under this stabilizing policy.
Due to flow conservation and link constraints, we have
\begin{equation}
\label{eq:F-nonnegative}
F_{ab}(t) \geq 0,
\end{equation}
\begin{equation}
\label{eq:F-conservation}
X_a(t) - Q_a(t) = \displaystyle{\sum_b F_{ab}(t) - \sum_c F_{ca}(t)},
\end{equation}
\begin{equation}
\label{eq:F-constraint}
F_{ab}(t) = \left\{
\begin{array}{ll}
0, &\hspace{-2mm}\hbox{if } (\R_a, \R_b) \not\in \M,\\
\hspace{-2mm}\displaystyle{\int_{\tau=0}^{t} \hspace{-3mm}\xi_i(p_i(\tau), \R_a, \R_b, z_i(\tau)) d\tau}, &\hspace{-2mm}\hbox{if } (\R_a, \R_b) \in \M_i
\end{array} \right.
\end{equation}
%
for all $a,b \in \{1, \ldots, N\}$
where $p_i(\tau)$ and $z_i(\tau)$ are the phase and traffic state, respectively, of junction $\J_i$ at time $\tau$.

For each $a,b \in \{1, \ldots, N\}$, define $f_{ab} \triangleq F_{ab}(\tilde{t})/\tilde{t}$ for some arbitrarily large time $\tilde{t}$.
It is clear from (\ref{eq:F-nonnegative}) and (\ref{eq:F-constraint}) that (\ref{eq:f-nonnegative}) and (\ref{eq:f-zero}) are satisfied.
In addition, we can follow the proof in \cite{NMR05} to show that there exists a sample paths $F_{ab}(t)$
such that $f_{ab}$ comes arbitrarily close to satisfying (\ref{eq:f-conservation}) and (\ref{eq:f-constraint}).
As a result, it can be shown that $\boldsymbol{\lambda}$ is a limit point of the capacity region $\Lambda$.
Since $\Lambda $ is compact and hence contains its limit points, it follows that $\boldsymbol{\lambda} \in \Lambda$.

Next, we show that $\boldsymbol{\lambda}$ strictly interior to $\Lambda$ is a sufficient condition for network stability,
considering only strategies that do not have a-priori knowledge of future events.
Suppose the rate vector $\boldsymbol{\lambda}$ is such that there exists $\boldsymbol{\epsilon} > 0$ such that
$\boldsymbol{\lambda} + \boldsymbol{\epsilon} \in \Lambda$.
Let $\mathbf{G} \in \Gamma$ be a transmission rate vector associated with the input rate vector $\boldsymbol{\lambda} + \boldsymbol{\epsilon}$
according to the definition of $\Lambda$.
It has been proved in \cite{NMR05} that there exists a stationary randomized policy
$\tilde{p}_i(\tau)$ for each $i \in \{1, \ldots, L\}$ that satisfies certain convergence bounds and such that
for each $(\R_a, \R_b) \in \M_i$,
$\lim_{t \to \infty} \frac{1}{t} \sum_{\tau=0}^{t} \xi_i(\tilde{p}_i(\tau), \R_a, \R_b, z_i(\tau)) = G_{ab}$.
In addition, such a policy stabilizes the system.
\end{proof}

\begin{corollary}
\label{cor:randomized}
Suppose 
$\mathbf{z}(t)$ is i.i.d. from slot to slot.
Then, $\boldsymbol{\lambda}$ is within the capacity region $\Lambda$
if and only if there exists a stationary randomized control algorithm that makes phase decisions $\hat{\mathbf{p}}(t)$ 
based only on the current traffic state $\mathbf{z}(t)$, and that yields for all $a \in \{1, \ldots, N\}$, $t \in \naturals$,
\begin{equation}
\begin{array}{c}
\expect \Bigg\{ V^{out}_a \big(\hat{\mathbf{p}}(t), \mathbf{z}(t) \big) 
 - V^{in}_a \big(\hat{\mathbf{p}}(t), \mathbf{z}(t) \big) \Bigg\}
= \lambda_a,
\end{array}
\end{equation}
where the expectation is taken with respect to the random traffic state $\mathbf{z}(t)$ and the (potentially) random control action based on this state.
\end{corollary}

Finally, based on the above corollary and the basic property of our backpressure-based traffic signal control algorithm,
we can conclude that our algorithm leads to maximum network throughput.

\begin{theorem}
If there exists $\boldsymbol{\epsilon} > 0$ such that $\boldsymbol{\lambda} + \boldsymbol{\epsilon} \in \Lambda$,
then the proposed backpressure-based traffic signal controller stabilizes the network, provided that
$\mathbf{z}(t)$ is i.i.d. from slot to slot.
\end{theorem}
\begin{proof}
Consider an arbitrary policy $\tilde{\mathbf{p}}(t)$.
By simple manipulations, we get
\begin{equation*}
\hspace{-2mm}
\begin{array}{l}
L(\mathbf{Q}(t+1)) - L(\mathbf{Q}(t)) \leq B -  \\ 
\hspace{3mm}2
\displaystyle{\sum_{a} Q_a(t) \Big( V^{out}_a \big( \tilde{\mathbf{p}}(t), \mathbf{z}(t) \big) - A_a(t) - V^{in}_a \big( \tilde{\mathbf{p}}(t), \mathbf{z}(t) \big) \Big)},
\end{array}
\end{equation*}
where $A_a(t)$ is the number of vehicle that exogenously enter the network at link $\R_a$ during time slot $t$,
\begin{equation*}
\begin{array}{rcl}
B &=&
\displaystyle{\sum_a \Bigg( \Big( \sup_{\scriptsize \begin{array}{c}\mathbf{p} \in \P,\\ \mathbf{z} \in \Z\end{array}}  
V^{out}_a \big( \mathbf{p}(t), \mathbf{z}(t) \big) \Big)^2} +\\
&&\displaystyle{\Big( A_a^{max} + 
\sup_{\scriptsize \begin{array}{c}\mathbf{p} \in \P,\\ \mathbf{z} \in \Z\end{array}}  
V^{in}_a \big( \mathbf{p}(t), \mathbf{z}(t) \big) \Big)^2 \Bigg)}
\end{array}
\end{equation*}
and $A_a^{max}$ satisfies $A_a(t) \leq A_a^{max}, \forall t$. Hence, we get
\begin{equation*}
\begin{array}{l}
\expect \Big\{ L(\mathbf{Q}(t+1)) - L(\mathbf{Q}(t)) \Big| \mathbf{Q}(t) \Big\} \leq\\
\hspace{8mm} \displaystyle{B + 2\sum_a Q_a(t) \expect \Big\{ A_a(t) \Big| \mathbf{Q}(t) \Big\} - 2 \sum_{a} Q_a(t)}\\
\hspace{8mm} \expect\Big\{ V^{out}_a \big( \tilde{\mathbf{p}}(t), \mathbf{z}(t) \big) - V^{in}_a \big( \tilde{\mathbf{p}}(t), \mathbf{z}(t) \big) \Big| \mathbf{Q}(t)\Big\}
\end{array}
\end{equation*}

However, from Lemma \ref{lem:basic_prop}, the proposed backpressure-based traffic signal controller minimizes
the final term on the right hand side of the above inequality over all possible alternative policies $\tilde{\mathbf{p}}(t)$.
But since $\boldsymbol{\lambda} + \boldsymbol{\epsilon} \in \Lambda$, according to Corollary \ref{cor:randomized},
there exists a stationary randomized algorithm that makes phase decisions based only on
the current traffic state $\mathbf{z}(t)$ and that yields for all $a \in \{1, \ldots, N\}$, $t \in \naturals$,
\begin{equation*}
\begin{array}{c}
\expect \Big\{ V^{out}_a \big( \tilde{\mathbf{p}}(t), \mathbf{z}(t) \big) - V^{in}_a \big( \tilde{\mathbf{p}}(t), \mathbf{z}(t) \big) 
 \Big| \mathbf{Q}(t) \Big\}
= \lambda_a + \epsilon.
\end{array}
\end{equation*}
Hence, we get that when the proposed backpressure-based traffic signal controller is used,
\begin{equation*}
\expect \Big\{ L(\mathbf{Q}(t+1)) - L(\mathbf{Q}(t)) \Big | \mathbf{Q}(t) \Big\} \leq B - 2\epsilon\sum_a Q_a(t),
\end{equation*}
and from Proposition \ref{prop:LyapunovStability}, we can conclude that the network is stable.
\end{proof}

\section{Simulation Results}
\label{sec:results}

First, we consider a 4-phase junction with 4 approaches and 8 links as shown in Figure \ref{fig:sim-setup}.
Vehicles exogenously entering each of the 8 links are simulated based on the data	Êcollected	from	the loop detectors installed	
at the junction between Clementi Rd and Commonwealth Ave W, Singapore.
The maximum output rate of each lane is assumed to be 4 times of the maximum arrival rate of that lane.

We implemented SCATS, which is the system currently implemented in Singapore, 
and our algorithm in MATLAB.
The parameters used in the SCATS algorithm are obtained from \cite{Liu03Master}.
Based on \cite{Pecherkova08Application,Pecherkova08Modelling}, the queue length on each link $\R_a$ evolves as follows.
\begin{equation}
Q_a(t+1) = Q_a(t) + I_a(t) - I_a^{\pi}(Q_a(t), I_a(t), R_a(t)),
\end{equation}
where $I_a(t)$ is the number of vehicles arriving at link $\R_a$ during time slot $t$ and
$I_a^{\pi}$ is a function that describes the number of passing vehicles and is given by
\begin{equation}
\label{eq:sim_rate}
I_a^{\pi}(Q_a(t), I_a(t), R_a(t)) = R_a(t)\left(1 - e^{\frac{-(Q_a(t) + I_a(t))}{R_a(t)}} \right).
\end{equation}
Here, $R_a(t) = S_a(t)g_a(t)$ is the maximum number of passing vehicles where $S(t)$ is the saturation flow and 
$g(t)$ is the green time for link $\R_a$.
\begin{figure}[h] 
   \centering 
        \includegraphics[trim=3cm 12cm 18cm 3cm, clip=true, width=0.11\textwidth]{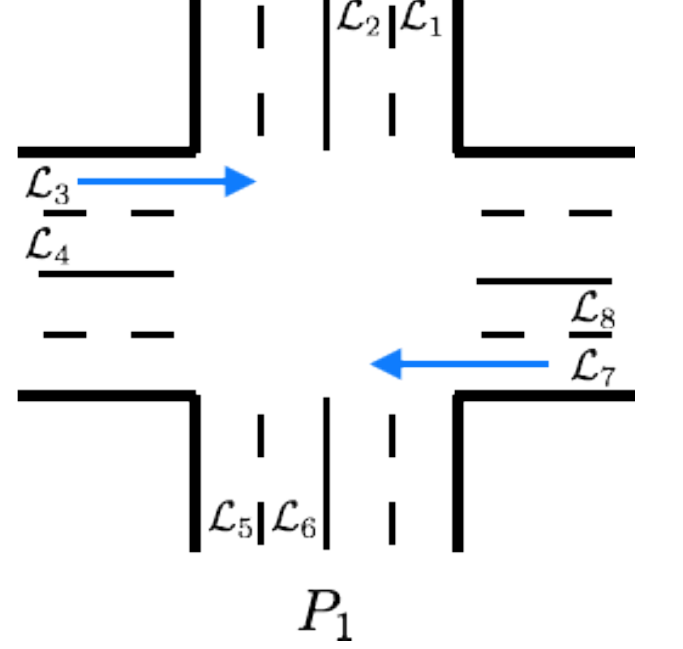}
        \hfill
        \includegraphics[trim=3cm 12cm 18cm 3cm, clip=true, width=0.11\textwidth]{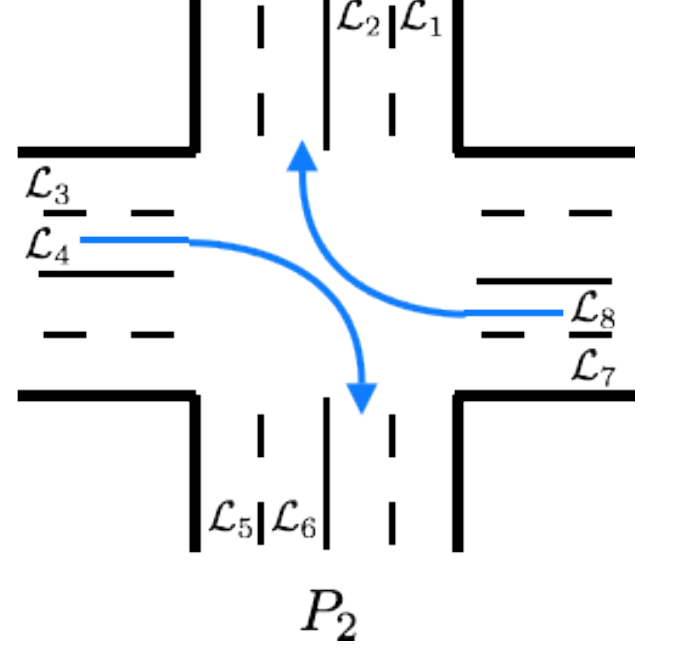}
        \hfill
        \includegraphics[trim=3cm 12cm 18cm 3cm, clip=true, width= 0.11\textwidth]{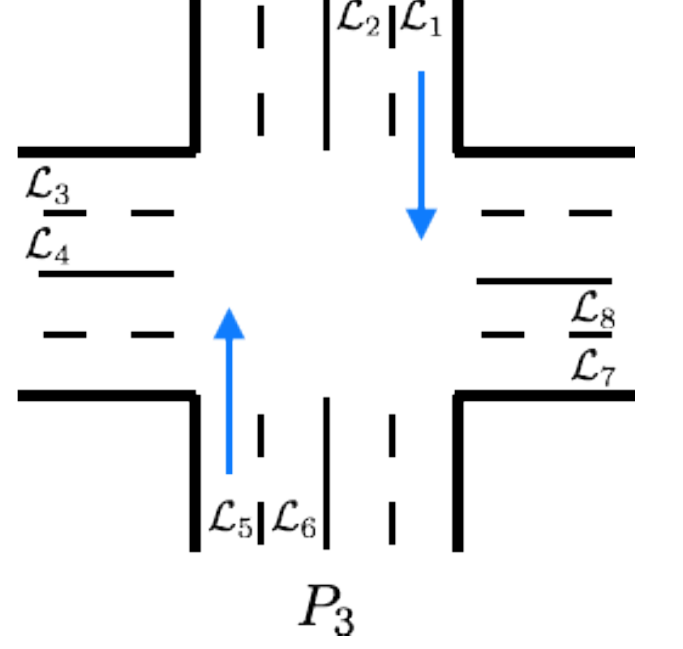}
        \hfill
        \includegraphics[trim=3cm 12cm 18cm 3cm, clip=true, width= 0.11\textwidth]{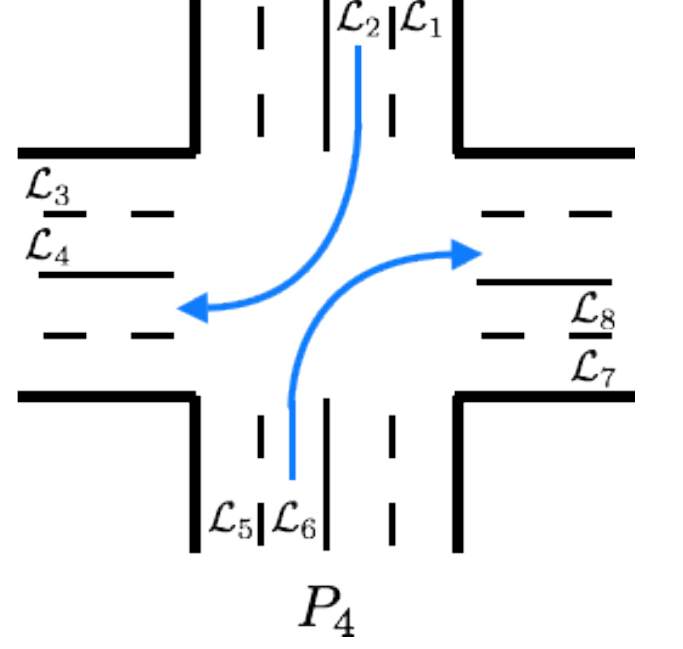}
   \caption{A 4-phase junction with 4 approaches 8 links used in our simulation.}
  \label{fig:sim-setup}
\end{figure}

Assuming that all the links have infinite queue capacity, 
queue lengths of each lane when our algorithm and SCATS are applied are shown in Figure \ref{fig:sim-results1}.
These simulation results show that our algorithm can reduce the maximum queue length by an order of magnitude, compared to SCATS,
as shown in Figure \ref{fig:sim-results-max1}.
Figure \ref{fig:sim-results-ave1} shows that our algorithm also performs significantly better on average.

\begin{figure}[h] 
   \centering 
        \includegraphics[trim=3cm 0.5cm 1.5cm 0.5cm, clip=true, width=0.45\textwidth]{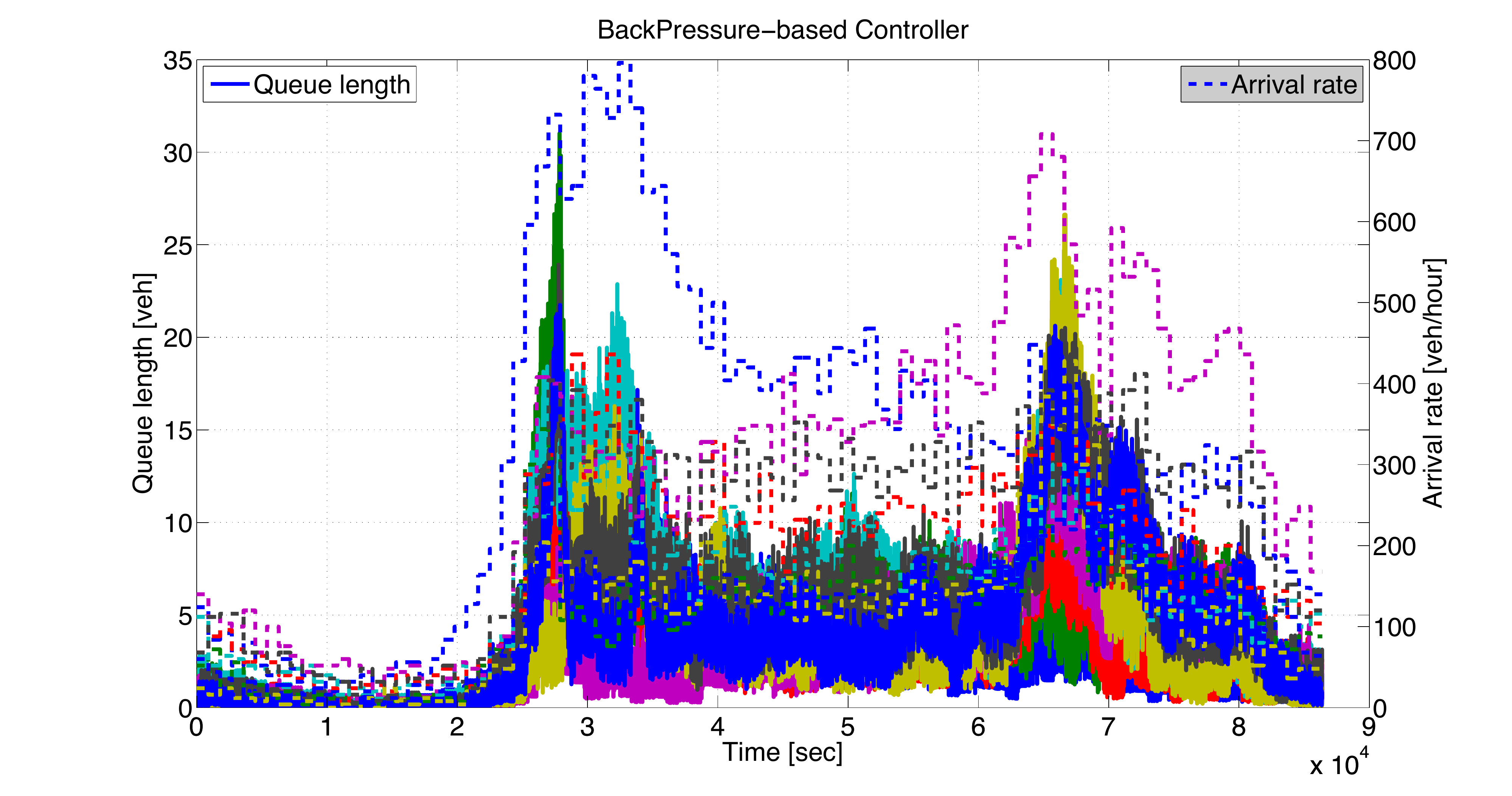}
        \includegraphics[trim= 3cm 0.5cm 1.5cm 0.5cm, clip=true, width=0.45\textwidth]{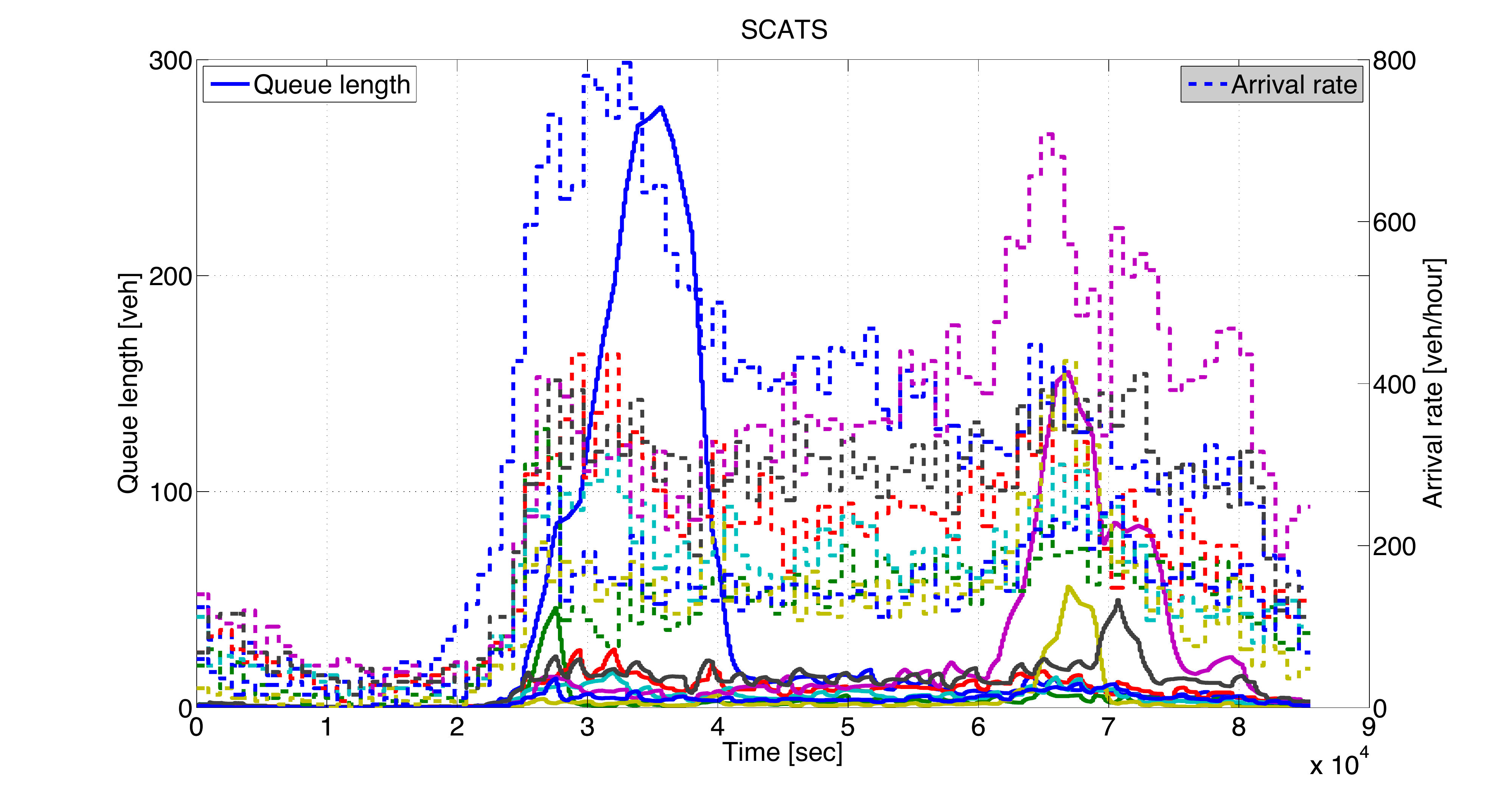}
   \caption{Simulation results showing the arrival rate (dashed line) and the resulting queue length (solid line) of each lane when 
   (top) backpressure-based controller and (bottom) SCATS are applied. Different colors correspond to different lanes.}
  \label{fig:sim-results1}
\end{figure}

\begin{figure}[h] 
   \centering 
        \includegraphics[trim=3cm 0.5cm 1.5cm 0.5cm, clip=true, width=0.45\textwidth]{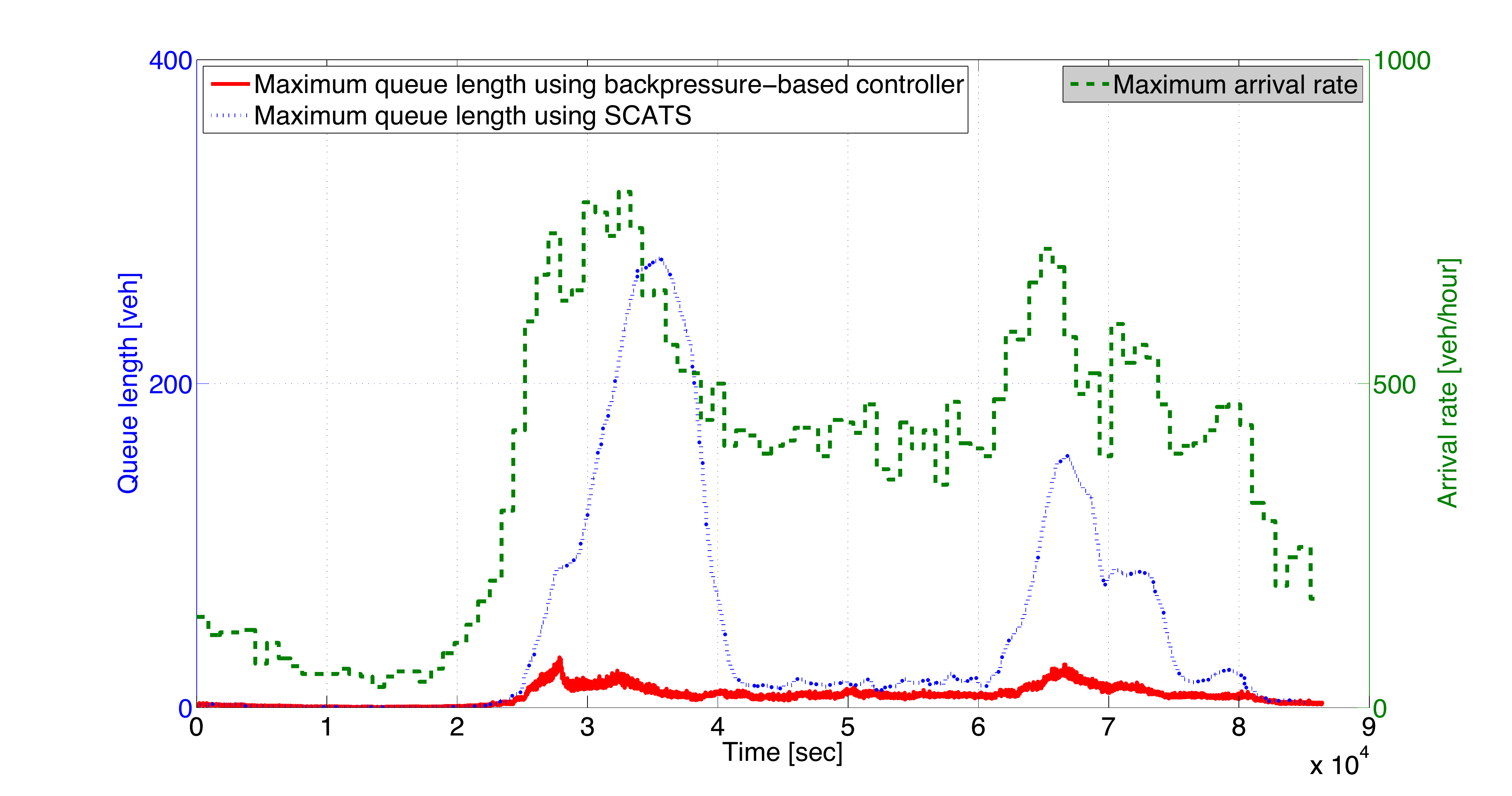}
   \caption{The maximum arrival rate and the maximum queue length over all the lanes when the backpressure-based controller and SCATS are applied.}
  \label{fig:sim-results-max1}
\end{figure}

\begin{figure}[h] 
   \centering 
        \includegraphics[trim=3cm 0.5cm 1.5cm 0.5cm, clip=true, width=0.45\textwidth]{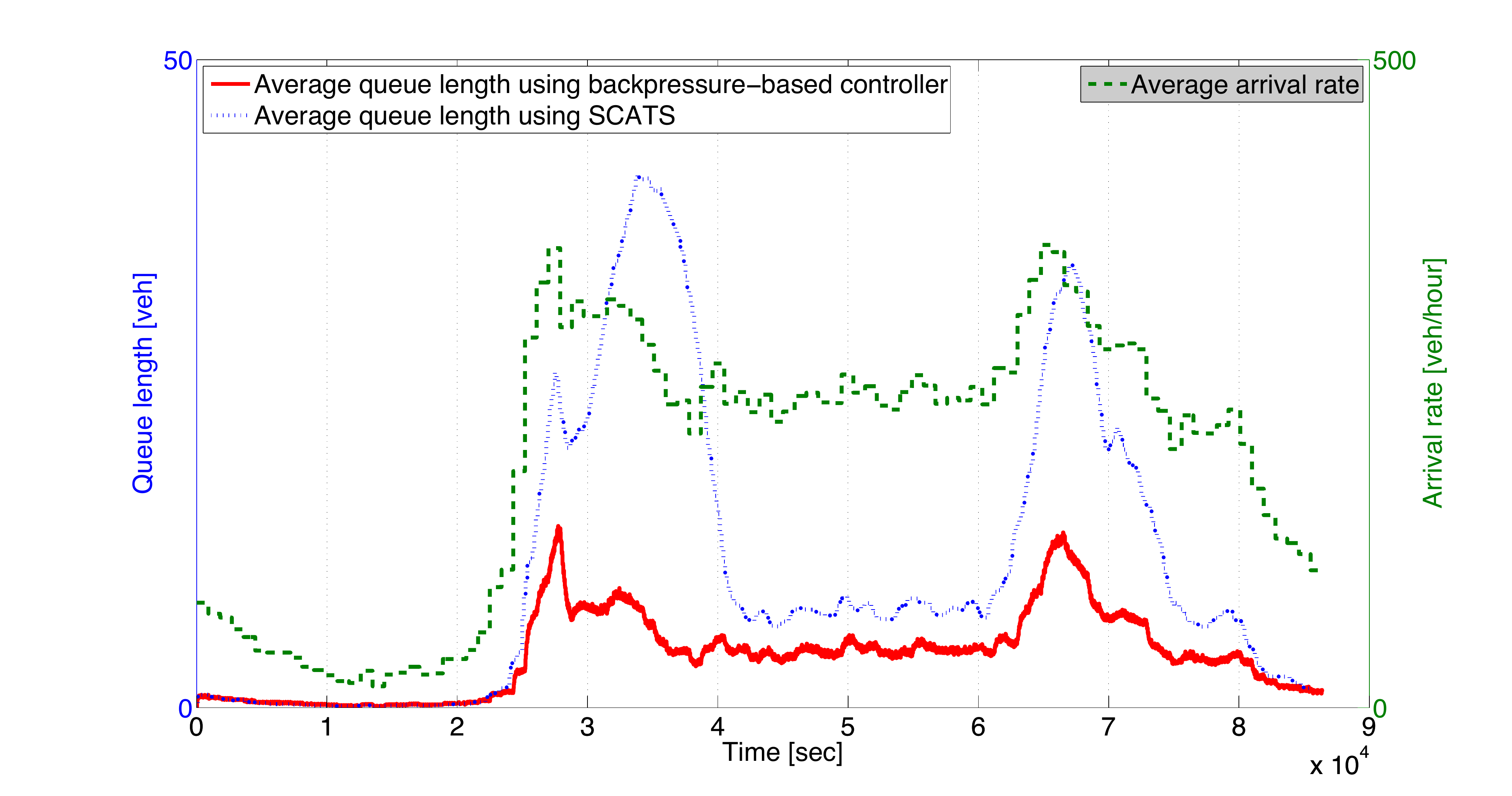}
   \caption{The average arrival rate and the average queue length over all the lanes when the backpressure-based controller and SCATS are applied.}
  \label{fig:sim-results-ave1}
\end{figure}

Suppose each link can actually accommodate only 100 vehicles. Figure \ref{fig:sim-results2} shows that 
SCATS can only support up to 0.9 times of the current vehicle arrival rate whereas
the the backpressure-based controller can support up to 1.3 times of the current vehicle arrival rate before the queue length exceeds the link capacity.

\begin{figure}[h] 
   \centering 
        \includegraphics[trim=3cm 0.5cm 1.5cm 0.5cm, clip=true, width=0.45\textwidth]{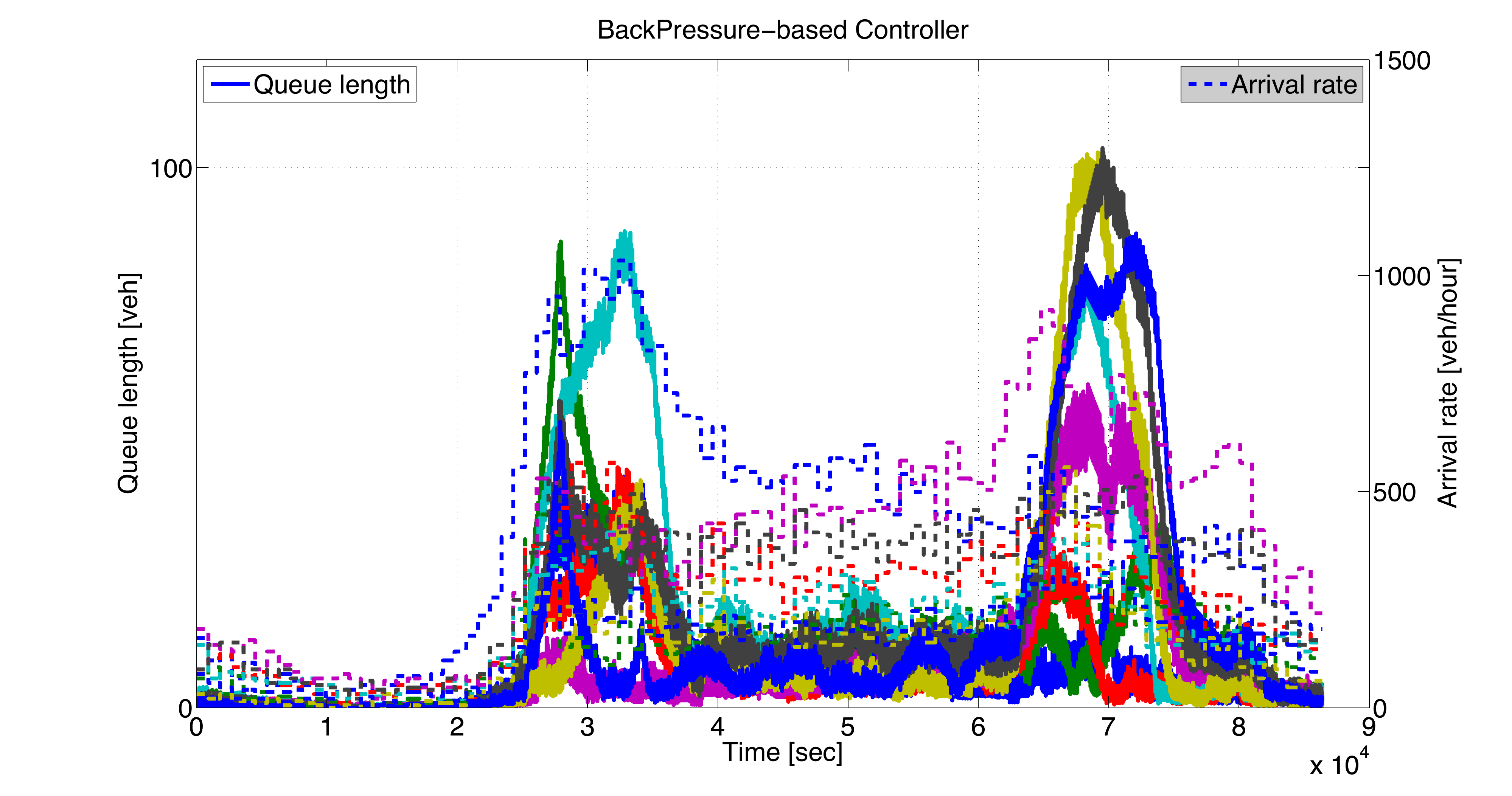}
        \includegraphics[trim= 3cm 0.5cm 1.5cm 0.5cm, clip=true, width=0.45\textwidth]{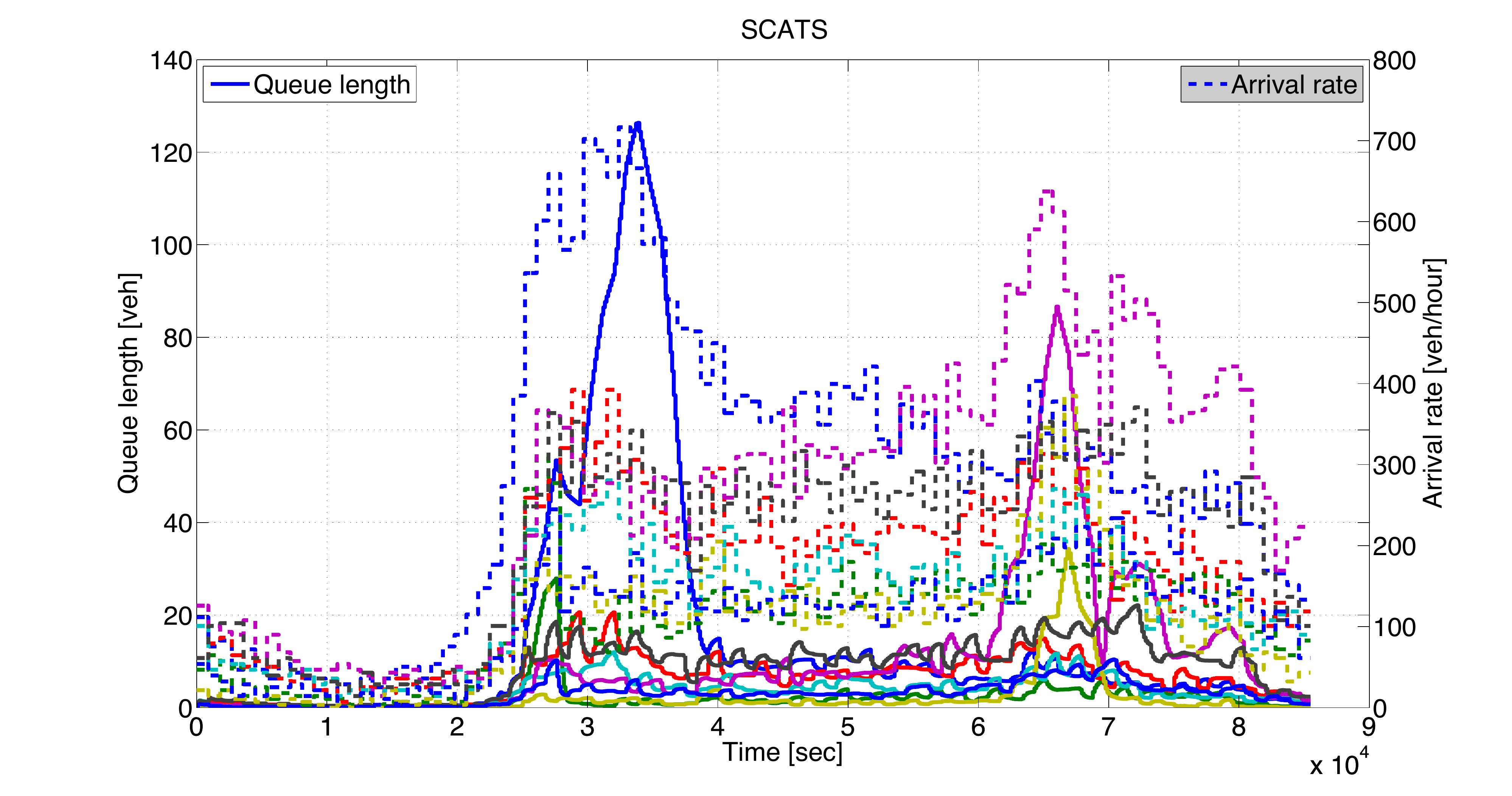}
   \caption{Simulation results showing the queue length (solid line) when 
   (top) the backpressure-based controller is applied with the vehicle arrival rate (dashed line) that is 1.3 times of the current value and 
   (bottom) SCATS is applied with the vehicle arrival rate (dashed line) that is 0.9 times of the current value. 
   Different colors correspond to different lanes.}
  \label{fig:sim-results2}
\end{figure}

Next, we employ a microscopic traffic simulator MITSIMLab \cite{Ben-Akiva01}, 
whose simulation models have been validated against traffic data collected from Swedish cities,
to evaluate our backpressure-based traffic signal control algorithm.
We consider a road network with 112 links and 14 signalized junctions as shown in Figure \ref{fig:sim-MITSIM-network}.
Vehicles exogenously enter and exit the network at various links based on 46 different origin-destination pairs,
with the arrival rate of 9330 vehicles/hour.
We implement SCATS and our backpressure-based traffic signal control algorithm in the traffic management simulator component of MITSIMLab.
Queue length (i.e., the number of vehicles) 
on each link when each algorithm is used is continuously recorded.
Note that in this case, the rate function $\xi_i$, which is used in our algorithm, is still derived from the macroscopic model in 
(\ref{eq:sim_rate}).
Hence, it may not accurately give the flow rate through the corresponding junction
due to a possible mismatch between the macroscopic model in (\ref{eq:sim_rate}) 
and the microscopic model used in MITSIMLab.
In addition, as opposed to the previous 1-junction case, all the links have finite queue capacity in this case.

\begin{figure}[h] 
   \centering 
        \includegraphics[trim=5.5cm 2.7cm 1.5cm 1.6cm, clip=true, width=0.48\textwidth]{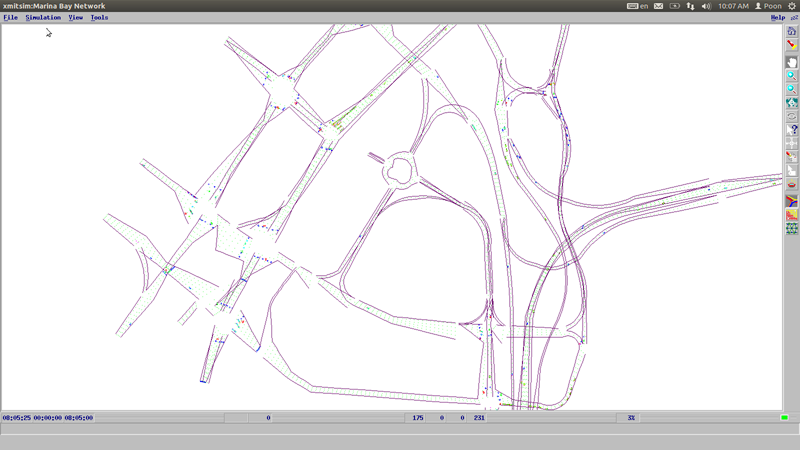}
   \caption{Road network used in the MITSIMLab simulation.}
  \label{fig:sim-MITSIM-network}
\end{figure}

The maximum and average queue lengths are shown in Figure \ref{fig:sim-results-MITSIM-max} and Figure \ref{fig:sim-results-MITSIM-avg}, respectively.
These simulation results show that our algorithm can reduce the maximum queue length by a factor of 3, compared to SCATS.
In addition, it performs significantly better on average.
One of the reasons that the difference in the queue lengths when our algorithm and SCATS are applied is not as significant as in
the previous 1-junction case is because in this case, each link has a finite capacity.
Hence, the number of vehicles on each link is limited by the link capacity and therefore queue length on each link cannot grow very large.
In fact, as shown in Figure \ref{fig:sim-MITSIM-spillback}, 
queue spillback, where queues extend beyond one link upstream from the junction, persists throughout the simulation,
especially when SCATS is used.

\begin{figure}[h] 
   \centering 
        \includegraphics[trim=3cm 0.5cm 3cm 0.5cm, clip=true, width=0.48\textwidth]{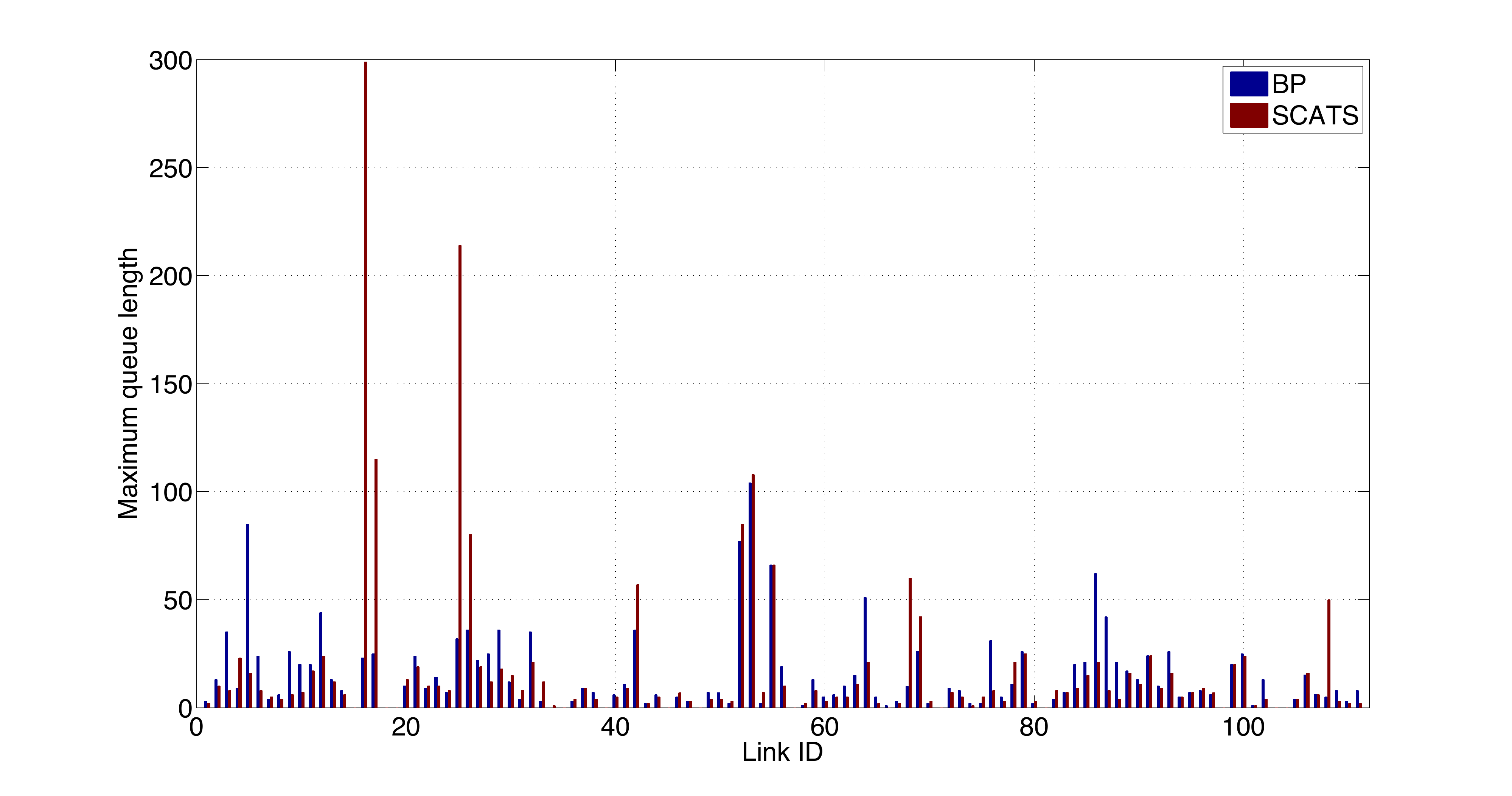}
        \includegraphics[trim= 3cm 0.5cm 3cm 0.5cm, clip=true, width=0.48\textwidth]{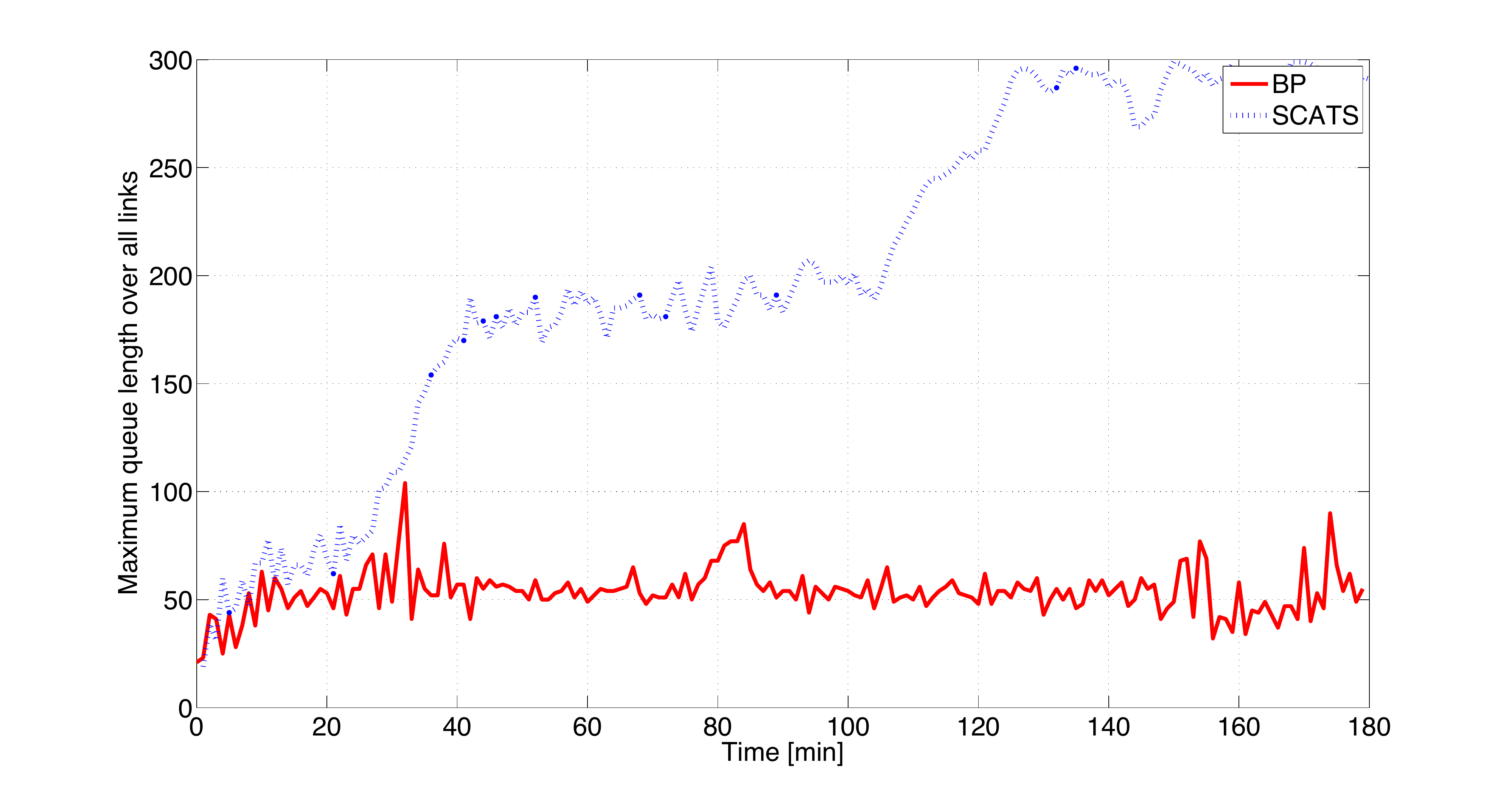}
   \caption{Simulation results showing maximum queue lengths when SCATS and our backpressure-based traffic signal control algorithm (BP) are used.}
  \label{fig:sim-results-MITSIM-max}
\end{figure}

\begin{figure}[h] 
   \centering 
        \includegraphics[trim=3cm 0.5cm 3cm 0.5cm, clip=true, width=0.48\textwidth]{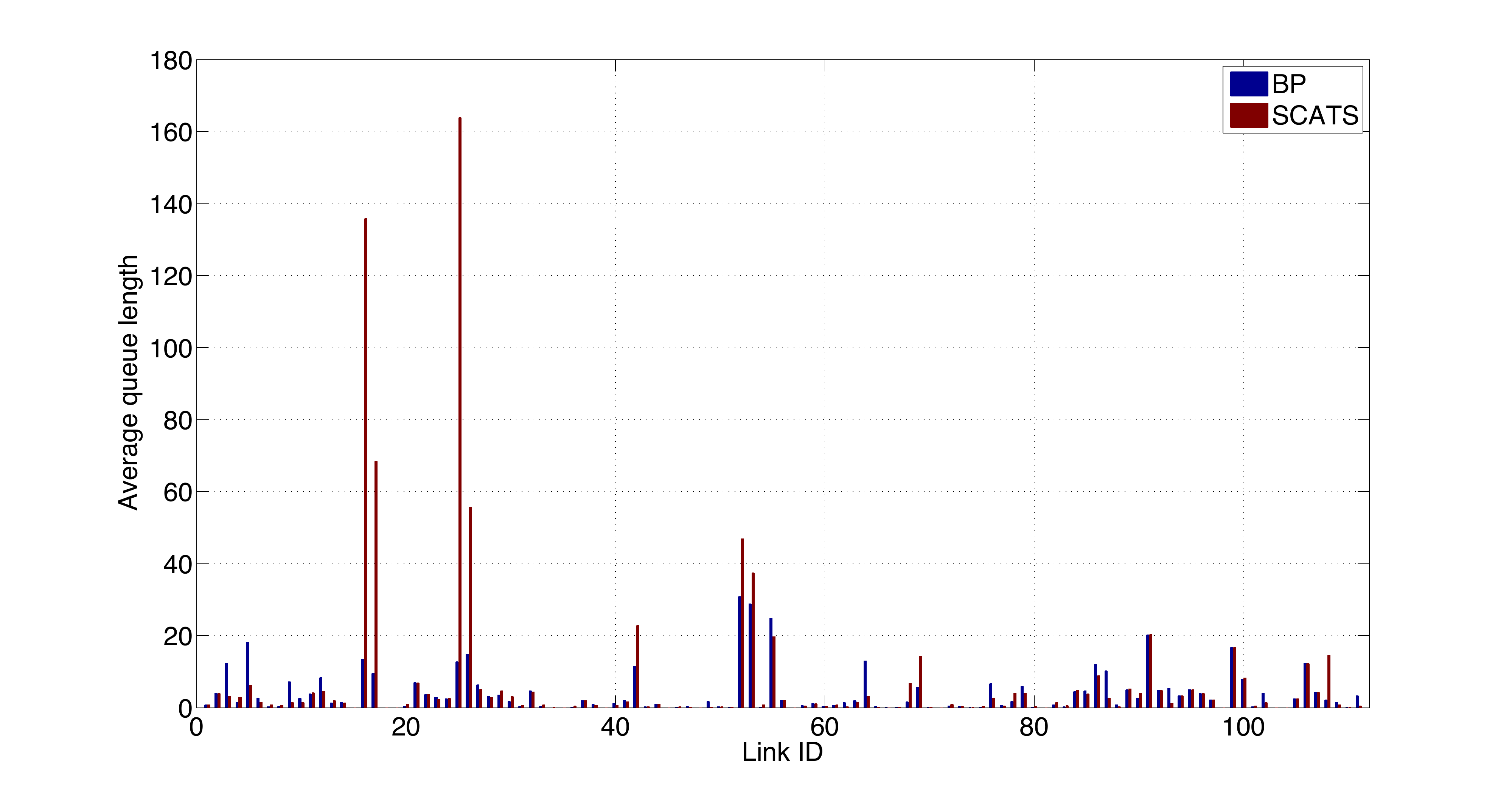}
        \includegraphics[trim= 3cm 0.5cm 3cm 0.5cm, clip=true, width=0.48\textwidth]{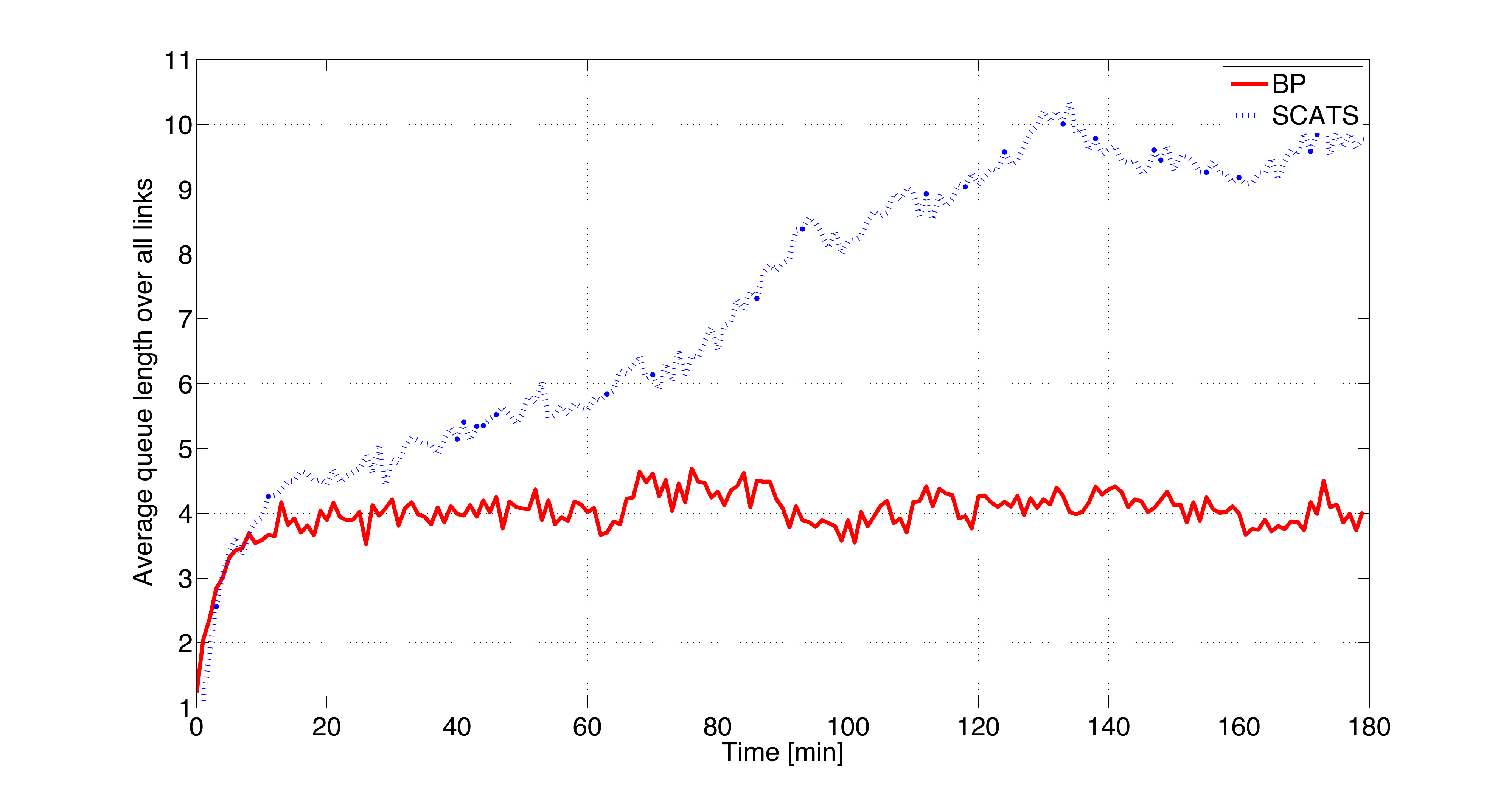}
   \caption{Simulation results showing average queue lengths when SCATS and our backpressure-based traffic signal control algorithm (BP) are used.}
  \label{fig:sim-results-MITSIM-avg}
\end{figure}

\begin{figure}[h] 
   \centering 
        \includegraphics[trim=16.5cm 2.7cm 20.5cm 1.6cm, clip=true, width=0.2\textwidth]{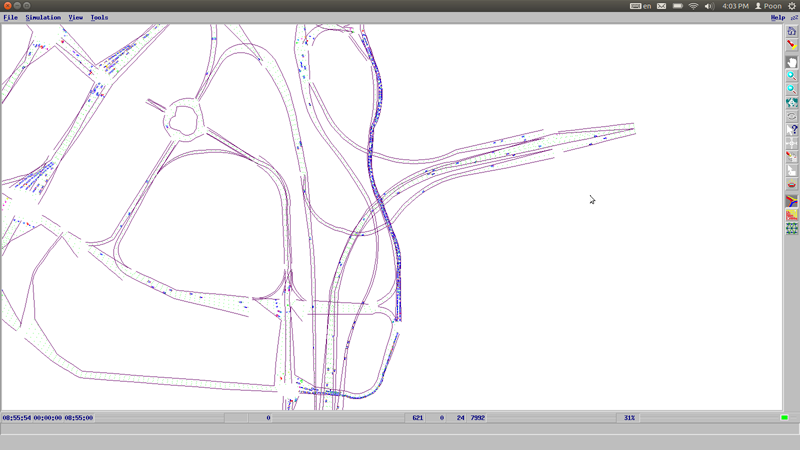}
        \hfill
        \includegraphics[trim=16.5cm 2.7cm 20.5cm 1.6cm, clip=true, width=0.2\textwidth]{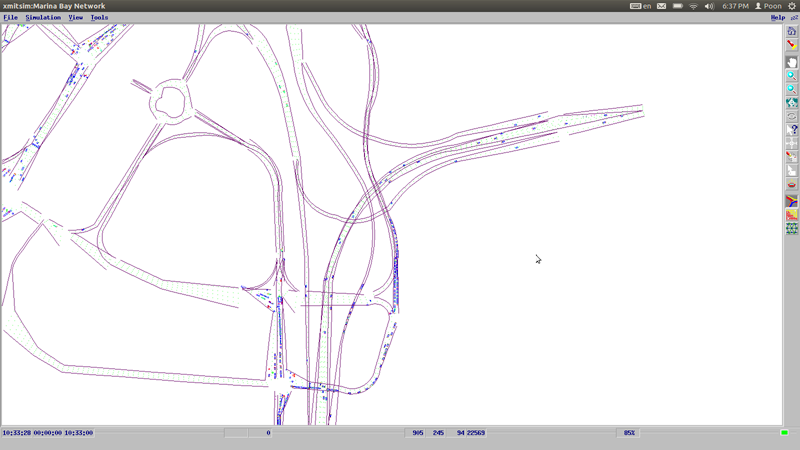}
   \caption{
   (left) Queues spread over multiple links upstream when SCATS is used, and
   (right) Queues do not spread over as many links when our backpressure-based traffic signal control algorithm is used.
   The part of the road that is filled with blue is occupied by vehicles.}
  \label{fig:sim-MITSIM-spillback}
  \vspace{-3mm}
\end{figure}

\section{Conclusions and Future Work}
\label{sec:conclusions}
We considered distributed control of traffic signals.
Motivated by backpressure routing, which has been mainly applied to communication and power networks,
our approach relies on constructing a set of local controllers, each of which is associated with each junction.
These local controllers are constructed and implemented independently of one another.
Furthermore, each local controller does not require the global view of the road network.
Instead, it only requires information that is local to the junction with which it is associated.
We formally proved that our algorithm leads to maximum network throughput even though
the controller is constructed and implemented in such a distributed manner
and no information about traffic arrival rates is provided.
Simulation results showed that our algorithm performs significantly better than SCATS, 
an adaptive traffic signal control systems that is being used in many cities.

Future work includes incorporating fairness constraints such as
ensuring that each traffic flow is served within a certain service interval.
Another issue that needs to be addressed as our algorithm may not lead
to periodic switching sequences of phases is the additional delay in drivers' responses to traffic signals,
unless a prediction of the next phase can be provided.
We are also investigating the coordination issue such as ensuring the emergence of green waves.

\section{ACKNOWLEDGMENTS}
The authors gratefully acknowledge Ketan Savla for the inspiring discussions, Prof. Moshe Ben-Akiva and his research group, in particular
Kakali Basak and Linbo Luo, for support with MITSIMLab,
and Land Transport Authority of Singapore for providing the data collected from the loop detectors
installed at the junction between Clementi Rd and Commonwealth Ave W.
This work is supported in whole or in part by the Singapore National Research Foundation (NRF) through the Singapore-MIT Alliance for Research and Technology (SMART) Center for Future Urban Mobility (FM).

\bibliographystyle{IEEEtran}
\bibliography{ref}
\end{document}